\documentclass{llncs}
\usepackage{amssymb}
\let\doendproof\endproof
\renewcommand\endproof{~\hfill\qed\doendproof}

\def\nicefrac{\frac}

\begin{document}

\title{Universal Factor Graphs
%\thanks{Work supported in part by The Israel Science Foundation (grant No. 873/08).}
}

\author{Uriel Feige
%\thanks{Department of Computer Science and Applied Mathematics, Weizmann Institute of Science, Rehovot, Israel. %\tt{uriel.feige@weizmann.ac.il}}~
and
Shlomo Jozeph}
%\thanks{Department of Computer Science and Applied Mathematics, Weizmann Institute of Science, Rehovot, Israel. %\tt{shlomo.jozeph@weizmann.ac.il}}}

\institute{
Department of Computer Science and Applied Mathematics,\\
The Weizmann Institute of Science, Rehovot, Israel\\
\email{\{uriel.feige,shlomo.jozeph\}@weizmann.ac.il}
\\ \bigskip \today
}

\maketitle

\begin{abstract}
The factor graph of an instance of a symmetric constraint satisfaction problem on $n$ Boolean variables and $m$ constraints (CSPs such as k-SAT, k-AND, k-LIN) is a bipartite graph describing which variables appear in which constraints. The factor graph describes the instance up to the polarity of the variables, and hence there are up to $2^{km}$ instances of the CSP that share the same factor graph. It is well known that factor graphs with certain structural properties make the underlying CSP easier to either solve exactly (e.g., for tree structures) or approximately (e.g., for planar structures). We are interested in the following question: is there a factor graph for which if one can solve every instance of the CSP with this particular factor graph, then one can solve every instance of the CSP regardless of the factor graph (and similarly, for approximation)? We call such a factor graph {\em universal}. As one needs different factor graphs for different values of $n$ and $m$, this gives rise to the notion of a family of universal factor graphs.

We initiate a systematic study of universal factor graphs, and present some results for max-$k$SAT. Our work has connections with the notion of preprocessing as previously studied for closest codeword and closest lattice-vector problems, with proofs for the PCP theorem, and with tests for the long code. Many questions remain open.
\end{abstract}

\section{Introduction}

A constraint satisfaction problem (CSP) has a set of $n$ variables
and a set of $m$ constraints (also referred to as clauses, or factors).
Every constraint involves a subset of the variables, and is satisfied
by some assignments to the variables and not satisfied by others.
An instance of a CSP is satisfiable if there is an assignment to the
variables that satisfies all constraints. When variables are Boolean
and constraints are symmetric a constraint is fully specified by the
set of literals that it contains (where a literal is either a variable
or its negation), and is satisfied if and only if the appropriate
number of literals is set to true (e.g., at least one for SAT, an
odd number for XOR, all for AND, the majority for MAJ, and at least
one but not all for NAE). To simplify the presentation, we shall consider
in this paper CSPs that are Boolean and symmetric, though we remark
that much of what we discuss can be extended to non-Boolean and non-symmetric
CSPs.

The \emph{factor graph} of an instance of a CSP is a bipartite graph.
Vertices on one side represent the variables, vertices on the other
side represent the constraints (also known as \emph{factors}), and
edges connect constraints to the variables that they contain. For
Boolean symmetric CSPs, a factor graph together with a labeling of
the edges with $\pm1$ (indicating whether the corresponding variable
has positive or negative polarity in the underlying clause) completely
specifies an instance of the CSP. Without the edge labels, there are
many instances of the CSP that share the same factor graph and differ
only in the polarity of the variables.

As is well known, deciding satisfiability for CSPs is NP-hard for
a large class of predicates (including, SAT, MAJ and NAE). See \cite{S78}
for a complete classification. Here we shall consider NP-hard CSPs.
The research question that motivates our current paper is to understand
what are the obstacles for obtaining efficient algorithms for solving
CSPs. Specifically, are algorithms having trouble in {}``understanding''
the structure of the factor graph, and this translates to difficulties
in solving the underlying CSP? Alternatively, are the computational
difficulties a result of the combinatorial richness of the polarities?

The structure of the factor graph may cause the underlying CSP instance
to be easy. For example, if the factor graph is a tree (or more generally,
of bounded treewidth), then the underlying CSP instance can be solved
in polynomial time (by dynamic programming). Our research question
(once properly formalized) can be viewed as asking whether in other
cases, the structure of the factor graph might be the major contributing
factor to making a CSP hard.

The playing field of our research agenda is greatly enriched once
optimization versions of CSPs are considered, namely max-CSP: find
an assignment to the variables that satisfies as many constraints
as possible. As is well known, even some polynomial time solvable
CSPs (such as XOR, or 2SAT) become NP-hard when their optimization
version is considered. See \cite{CKS01} for a classification. A standard
way of dealing with NP-hard max-CSP instances is via approximation
algorithms that in polynomial time find an assignment that is guaranteed
to satisfy a number of constraints that is at least $\rho$ times
the maximum number of constraints that can be satisfied, for some
$0<\rho<1$. For many CSPs, the best possible $\rho$ is known, in
the sense that the approximation ratios provided by known approximation
algorithms are matched by hardness of approximation results that show
that better approximation ratios would imply that P=NP. For example,
$\rho=7/8$ is a tight approximation threshold for max-3SAT \cite{H97}.
Moreover, for all CSPs, an algorithm (based on semidefinite programming)
with the optimal approximation ratio is given by Raghavendra \cite{R08}, assuming
the \emph{Unique Games Conjecture} of Khot \cite{K02}. However, despite
the optimality of this algorithm, it is difficult to figure out which
approximation ratio it guarantees, and consequently there are CSPs
for which the value of this threshold is not known. (And of course,
if the Unique Games Conjecture is false then the approximation ratio
implied by this algorithm need not be tight.)

Our research agenda naturally extends to max-CSP. One may ask whether
approximation algorithms are having trouble in {}``understanding''
the structure of the factor graph, and whether this translates to
difficulties in approximating the underlying CSP. Moreover, now the
question acquires also a quantitative aspect, and one may ask to what
extent does the factor graph contribute to the approximation difficulty.
For example, if algorithms had no difficulty in {}``understanding''
factor graphs, could the approximation ratio for max-3SAT be improved
from $7/8$ to $8/9$?

As in the case of tree factor graphs for decision versions, there
are known families of factor graphs (such as planar graphs, or more
generally, families of graphs excluding a fixed minor) on which the
underlying CSP instance has improved approximation ratios, or even
a PTAS ($\rho>1-\epsilon$ for every $\epsilon>0$). On the other
hand, it appears that for some CSPs, almost every factor graph is
difficult. For example, there is no known approximation algorithm
that runs in polynomial time on random 3CNF formulas (with say $m=n\log n$
constraints) and approximates max-3SAT within a ratio better than~7/8.
This suggests (though does not prove) that there is no need for clever
design of the factor graph in order to make the underlying CSP instance
difficult -- almost any factor graph would do.

In contrast, for unique games (which is a special family of CSPs with
two non-Boolean variables per constraint), the approximation ratios
achievable on random factor graphs \cite{AKKSTV08} are much better
than those currently known to be achievable on arbitrary factor graphs.
(Technically, the graphs considered by Arora et al.~\cite{AKKSTV08} have variables
as vertices and constraints as edges, but there is a one-to-one correspondence
between such graphs and factor graphs.) The same holds for some other
classes of graphs \cite{T05,K10}. Can we (and should we) identify
more factor graphs on which unique games are easy? Is there a {}``universal''
graph (e.g., a generalized Kneser constraint graph?) such that if
unique games are easy on it, then the Unique Games Conjecture is false?
Such questions lead naturally to the notion that we call here \emph{universal
factor graphs}.

\subsection{Preprocessing}

How can we provide evidence that algorithms for max-3SAT should be spending substantial time in analyzing the factor graph? Here is a possible formal approach. Reveal the input instance in two stages. In the first stage, only the factor graph is revealed. At this point the algorithm is allowed to run for arbitrary time and record (in polynomial space) whatever information about the factor graph that it may hope to find useful (e.g., an optimal tree decomposition of the factor graph, or a minimum dominating set in the factor graph, both of which are pieces of information that take exponential time to compute). Thereafter the polarities of the variables are revealed. At this stage the algorithm has only polynomial time, and it needs to find an optimal solution to the max-3SAT instance. If there is a combination of algorithms (unbounded time for stage~1, polynomial time for stage~2) that can do this on every instance, this establishes that a good understanding of the factor graph suffices for solving 3SAT instances. If this cannot be done, this establishes that at least some substantial portion of the running time is a result of the combinatorial richness of space of possibilities for polarities of the variables. Refined versions of the preprocessing approach either require less of the stage~2 algorithm (finding nearly optimal solutions rather than optimal ones) or give it extra power (allow subexponential time), and may lead to a more quantitative understanding of the value of preprocessing.

To derive positive results in this model, it suffices to provide the respective algorithms and their analysis. But how does one provide negative results? This is where the notion of {\em universal factor graphs} comes in. Informally, these are factor graphs on which preprocessing is unlikely to help, because if it does, then all instances (regardless of their factor graph) can be solved even without preprocessing.

\subsection{Universal Factor Graphs}

We consider infinite families of factor graphs. Basically, for every value of $N,M > 0$, a family includes at most one factor graph with $N$ variables and $M$ constraints. However, for convenience in intended future uses, members of the family are indexed by two auxiliary indices that are called $n$ and $m$. Definition~\ref{def:family} does not exclude the possibility that several factor graphs in the family share the same values of $N$ and $M$, but their number is upper bounded by some polynomial in $N + M$.

\begin{definition}
\label{def:family}
Consider an arbitrary CSP with $k$ variables per-constraint.
For integers $n > 0$ and $0 < m \le 2^k{n \choose k}$, let $N(n,m)$ and $M(n,m)$ be two functions, each lower bounded by $n$ and upper bounded by a polynomial in $n + m$. A {\em family of factor graphs} associates with each pair of values of $n$ and $m$ a factor graph with $N(n,m)$ variables and $M(n,m)$ constraints. The family is {\em uniform} if there is an algorithm running in time polynomial in $n + m$ that given $n,m$ produces the associated factor graph.
\end{definition}

Every member of a family of factor graphs for a $k$-CSP can give rise to $2^{kM}$ instances of the CSP, depending on how one sets the polarities of the variables in the constraints. Given any such instance as input, we shall consider computational tasks such as {\em satisfiability} (find a satisfying assignment if one exists), {\em optimization} (find an assignment satisfying as many clauses as possible) and {\em approximation} (get close to optimal).

The algorithms that perform the above tasks will be limited in their running times. In this work, we shall be interested in two classes of running times. One is the standard {\em polynomial time} (P) notion, which in our case will mean polynomial in $(N + M)$. The other is {\em subexponential time}, (SUBEXP) which in this paper is taken to mean time time $2^{O(N^{1 - \epsilon})}$ for some $\epsilon > 0$.

Recall that in computational complexity theory, one distinguishes between uniform
models of computation (such as Turing machines) and non-uniform models
(such as families of circuits). This
distinction is relevant in our context. The notion of preprocessing the factor graph can be captured by allowing for nonuniform algorithms. Hence we shall be dealing with the complexity classes P/poly, SUBEXP/poly and SUBEXP/subexp (the parameters /poly and /subexp correspond to the length of advice that the preprocessing stage is allowed to record). For simplicity in our presentation, in each of our definitions below we shall specify one particular complexity class (either P/poly or SUBEXP/poly), but we note that our results extend to other complexity classes as well (such as P instead of P/poly, or SUBEXP/subexp instead of SUBEXP/poly).

In this work we will show that for some uniform families of factor graphs solving satisfiability or approximation tasks are hard. These families of factor graphs will be referred to as {\em universal}, and with slight abuse of terminology, individual factor graphs within these families will be referred to as {\em universal factor graphs}. The hardness results will be proved under some complexity assumption. If the complexity assumption is widely believed, such as that NP is not contained in P/poly, then the universal factor graphs support the view that the complexity of the underlying CSP cannot be attributed entirely to the factor graph and is at least partly due to the polarities of the variables, because the nonuniform algorithms could preprocess the factor graph for arbitrary time prior to receiving the polarities of the variables. If the complexity assumption is not as widely believed (such as the Unique Games Conjecture), the interpretation of these hardness result can be that if one wishes to refute the complexity assumption, it would suffice to design algorithms that are specifically tailored to work on instances with factor graphs as in the universal family.

We now present formal definitions that are tailored to match those results that we can prove in this paper. It is straightforward to adapt these definitions to other variations as well.

\begin{definition}
\label{def:P-universal}
For a given CSP, a uniform family of factor graphs is P-universal if there is no P/poly algorithm for instances of the CSP with factor graphs from this family, unless NP is contained in P/poly.
\end{definition}

\begin{definition}
\label{def:subexp-universal}
For a given CSP, a uniform family of factor graphs is subexp-universal if there is no SUBEXP/poly algorithm for instances of the CSP with factor graphs from this family, unless there is a SUBEXP/poly algorithm for all instances of the CSP.
\end{definition}

\begin{definition}
\label{def:threshold-universal}
For a given CSP and $0 < \rho < 1$, a uniform family of factor graphs is $\rho$-universal if there is no P/poly approximation algorithm with approximation ratio better than $\rho$ on the instances of the CSP with factor graphs from this family, unless NP is contained in P/poly. This notion is referred to as {\em threshold-universal}. If $\rho$ is equal to the best approximation ratio known for the underlying CSP, we will refer to this as a {\em tight} threshold. When we do not wish to specify a particular value for $\rho$, we call the family {\em APX-universal}. A variation on $\rho$-universality is $(c,s)$-universality with $0 < s < c \le 1$, where instead of approximation within a ratio of $\rho$, one considers distinguishing between instances with at least a $c$-fraction of the clauses being satisfiable, and instances with at most $s$-fraction being satisfiable. For a CSP for which the decision
variant is NP-hard (e.g. 3SAT), $\rho$-universality will be taken to mean $\left(1,\rho\right)$-universal.
\end{definition}

More generally, for optimization versions we shall allow vertices
(representing constraints) of universal factor graphs to have nonnegative
weights, thus representing instances in which one wishes to find an
assignment that maximizes the weight (rather than the number) of satisfied
constraints. As the weights will be fixed (independently of the subsequent
polarities given to variables), this is in essence a condensed representation
of an unweighted universal factor graph (which can be obtained by
duplicating each vertex a number of times proportional to its weight,
rounded to the nearest integer -- details omitted).

\subsection{Some Research Goals}

The notion of universal factor graphs opens up many research directions that we find interesting.
In our current work we attempt to answer questions such as: Does 3SAT have P-universal factor graphs? Subexp-universal factor graphs? Does max-3SAT have APX-universal factor graphs? Does max-3SAT have
$7/8$-universal factor graphs?
These questions are part of a wider research agenda that concerns
questions such as: Do all CSPs have tight threshold-universal factor graphs? Which CSPs do not have tight threshold-universal factor graphs?
Other questions of interest include: How do universal factor graphs look like? Can knowledge of their structure help us either in designing new algorithms, or in reductions that
prove new hardness results?

\subsection{Related Work\label{sec:related}}

There has been work showing that CSPs on particular factor graphs
are NP-hard, and using such results to help in reductions establishing
further NP-hardness results. For example, it is known that 3SAT is
NP-hard even when the factor graph is planar \cite{L82}, and this
was used (for example) in showing that \emph{minimum-length rectangular
partitioning of a rectilinear polygon} (with holes) is NP-hard \cite{LPRS82}.
Our notion of universal factor graphs is stronger as it requires at
most one particular factor graph for each instance size, rather than
a whole family of factor graphs (e.g., the $n$ by $n$ grid, rather
than all planar graphs).

A line of work that closely relates to our research agenda is that
of preprocessing for NP-hard problems. As the universal factor graph
is fixed, one may consider preprocessing it for arbitrary (exponential)
time in order to produce a polynomial size {}``advice'', prior to
getting the polarities of the variables. Preprocessing was extensively
studied for some NP-hard problems, and hardness results in the context
of preprocessing amount to designing instances that are universal
(in our terminology). Naor and Bruck \cite{BN90} show that the nearest
code word problem remains NP-hard even when the code can be preprocessed.
Nearest lattice vector (CVP) when the lattice can be preprocessed
was shown to be NP-hard and APX-hard by Feige and Micciancio \cite{FM02}. The tightest hardness
results for lattice problems with preprocessing currently known are by Khot et al.~\cite{KPV}. An earlier work by Alekhnovich et al.~\cite{AKKV05} has some partial overlap with our current work, because it uses PCP theory and in the process gives hardness of approximation
results with preprocessing for additional problems. See more details in Section~\ref{sec:threshold}.

The above results on coding and lattice problems with preprocessing
are motivated by the fact that in these problems, it is indeed often
the case that part of the input is fixed in advance (the code, or
a basis for the lattice), and part of the input (a noisy word that
one wishes to decode, or vector for which one wishes to find the closest
lattice point) is a query that is received only later. Moreover, multiple
queries are expected to be received on the same fixed input. In these
cases it really makes sense to invest much time in preprocessing the
fixed part of the input, if this later helps answering the multiple
queries more quickly. In contrast, our notion of universal factor graphs is independent of such practical concerns. Our motivation is to understand the source of difficulties in solving NP-hard problems. In particular, it is irrelevant to us whether there really is any real life situation in which one receives the factor graph of a 3CNF formula in advance, and then is asked a sequence of queries about it, each time with different polarities of the variables.

Is it at all plausible that preprocessing can help? For lattice problems,
this indeed appears to be the case. There are no known approximation
algorithms with subexponential ratios for CVP, but if preprocessing
is allowed, than polynomial approximation ratios are known (by using
an exponential time preprocessing procedure that derives a so called
\emph{reduced basis} of the lattice). For CSPs, the authors are aware
of only much weaker evidence that preprocessing may help. This relates
to the case that polarities of variables are random rather than arbitrary.

%We elaborate on this in appendix \ref{sec:randompolarities}.

%\section{Universal factor graphs for random polarities\label{sec:randompolarities}}

There is a refutation algorithm that is poly-time on random 3CNF formulas
with more than $n^{1.5}$ clauses. The obstacle to extending this
to lower density of $n^{1.4}$ is graph-theoretic: if one knew how
to efficiently find certain substructures in the factor graphs (that
almost surely exist), this would suffice \cite{FKO06}. Preprocessing
the factor graph would allow finding these structures. Hence at these densities, random
factor graphs are not expected to be universal (with respect to \emph{random}
polarities).

%Moreover, it is plausible that this holds not only with
%respect to random factor graphs.
%\begin{conjecture}
%For 3CNF formulas with more than $n^{1.4}$ clauses, there are no
%difficult factor graphs - after preprocessing, refutation will be
%poly-time for random polarities.
%\end{conjecture}
%The above conjecture is true even without preprocessing when the clause
%density is larger. In \cite{F07} a refutation algorithm is shown
%to run in polynomial time (with high probability) on 3SAT instances
%with $\Omega(n^{1.5}\sqrt{\log\log n})$ clauses, regardless of the
%factor graph, and with probability taken only over the choice of polarities
%of the variables.

In the current paper we consider arbitrary polarities for the variables
rather than random polarities. Nevertheless, we remark that the case
of random polarities is also well motivated, and related to possible
cryptographic application. See \cite{ABW10} as an example showing
how results from \cite{FKO06} can be used in a proposal of new public
key cryptographic primitives.

More generally, cryptography offers many examples where preprocessing is believed to help (it will lead to the discovery of a so called {\em trapdoor} that would make solving future instances easy), but as this typically relates to computational problems that are believed not to be NP-hard, further discussion of this is omitted from the current manuscript.

\subsection{Our Results}

The first theorem is based on a straightforward reduction and we have no doubt that it was previously known (perhaps using different terminology).

\begin{theorem}\label{thm:polytime} There are P-universal factor
graphs for 3SAT. \end{theorem}

For the P-universal factor graphs constructed by our proof for Theorem~\ref{thm:polytime}, an algorithm running in time $2^{N^{1 - \epsilon}}$ on instances of the universal family would correspond to time $2^{n^{3 - 3\epsilon}}$ on general instances. Hence they are not subexp-universal. The next theorem addresses this issue.

\begin{theorem}\label{thm:subexp} There are subexp-universal factor
graphs for 3SAT. \end{theorem}

We would have liked to prove that there are $7/8$-universal factor
graphs for max-3SAT, matching the tight threshold of approximability
for max-3SAT. However, we only managed to prove weaker bounds.

\begin{theorem}\label{thm:threshold} There are $77/80$-universal factor
graphs for max-3SAT. \end{theorem}

Is there any CSP for which we can obtain tight threshold-universal families? We do not know, but we do have almost tight results.

\begin{theorem}
\label{thm:tight}
For every $\epsilon > 0$ there is an integer
$k$ for which there is a family
of factor graphs that are $\left(1-\left(1-\epsilon\right)2^{-k}\right)$-universal for max-E$k$SAT.
\end{theorem}

Theorem~\ref{thm:tight} in nearly tight because every instance of max-E$k$SAT is $(1 - 2^{-k})$-satisfiable, and consequently there are several algorithms with a $(1 - 2^{-k})$ approximation ratio. To actually get tight results we would need to switch the order of quantifiers in Theorem~\ref{thm:tight} (show that for some $k$ the result holds for every $\epsilon$), but doing so remains an open question.

Using the techniques developed in our work and known reductions among CSPs one can obtain APX-universal factor graphs for additional CSPs. In particular, we derive APX-universal factor graphs for max-2LIN, thus illustrating that for approximating unique games (max-2LIN is a unique game) at least part of the difficulty comes from the polarities of variables rather than from the structure of the factor graph. See Appendix~\ref{app:MoreCSPs}.

\section{Overview of proofs}
\label{sec:proofs}

At a high level, to show that a factor graph is universal, one shows
that any other factor graph (of the appropriate size) can be reduced
to it. The details of how this is done depend on the context.

The proof of Theorem \ref{thm:polytime} appears in  Appendix~\ref{sec:poly}. It is elementary
and can serve as an introduction to some of the more complicated proofs that follow.

\subsection{Subexp-Universal Families\label{sub:SubexpUFG}}

Our proof of Theorem~\ref{thm:subexp} combines two ingredients. One
is a variation on a result of Impagliazzo et al.~\cite{IPZ01} (see Lemma~\ref{lem:IPZ} in Appendix~\ref{app:SubexpUFG}). It can be leveraged to show that for the
purpose of constructing subexp-universal factor graphs it suffices
to consider 3CNF instances with a linear number of clauses.

The other ingredient is a reduction with a tighter connection between
$n+m$ and $N$ compared to the one used in our proof of Theorem \ref{thm:polytime}.
\begin{lemma}
\label{lem:polylog} There is a factor graph with $N=O(m\log m\log n)$
variables that is P-universal with respect to 3SAT instances
with $n$ variables and $m$ clauses.
\end{lemma}
Our proof of Lemma~\ref{lem:polylog} makes use of oblivious sorting
networks (specifically, the one of Ajtai et al.~\cite{AKS83}).

More details on those two ingredients and how they are combined to
prove Theorem \ref{thm:subexp} appear in appendix \ref{app:SubexpUFG}.

%\subsection{APX-universal families\label{sec:APX}}

\subsection{Threshold-Universal Families}
\label{sec:threshold}

%The overall structure of the proof of Theorem~\ref{thm:threshold} is based on techniques developed in %earlier work (such as Dinur's proof of the PCP theorem~\cite{D07,RS07}, and encodings using the long %code of Bellare et al.~\cite{BGS98}). However, some modifications are needed, and in particular our %proof introduces a new concept, that of {\em oblivious folding} of the long code.

For our proof of Theorem \ref{thm:threshold} we use a notion that we call a \emph{factor
graph preserving reduction} (FGPR). It is an algorithm that transforms
a \emph{source} 3CNF instance $f_{s}$ to a \emph{target} 3CNF instance
$f_{t}$. The transformation has the following properties:

\begin{enumerate}
\item Polynomiality. The transformation algorithm runs in polynomial time
(in the size of $f_{s}$). Consequently, the size of $f_{t}$ is polynomial
in the size of $f_{s}$.
\item Faithfulness. If $f_{s}$ is satisfiable, so is $f_{t}$, and vice versa.
\item Factor graph preserving. Any two instances $f_{s}$ and $f'_{s}$
with the same factor graph are reduced to two instances $f_{t}$ and
$f'_{t}$ that have the same factor graph.
\end{enumerate}

To be useful for our purposes, we would like the FGPR to also have a {\em gap amplification} aspect. Namely, if $f_s$ is not satisfiable, then the fraction of clauses satisfiable in $f_t$ is smaller than the fraction of clauses satisfiable in $f_s$.

Theorem \ref{thm:threshold} will be broken into two sub-theorems, each of which is proved using FGPRs.

\begin{theorem}\label{thm:APX} There are APX-universal factor graphs
for max-3SAT. \end{theorem}

\begin{theorem}\label{thm:threshold1} There is a reduction from
APX-universal factor graphs for max-3SAT to $77/80$-universal ones.
\end{theorem}

The proof of Theorem \ref{thm:APX} strongly relates to the work of
Alekhnovich et al.~\cite{AKKV05}. As explained in Section \ref{sec:related}, in that
work various APX-hardness results with preprocessing were obtained.
Among them, there were APX-hardness results with preprocessing for
certain CSPs (satisfying quadratic equations). It is not difficult
to use these results in order to obtain APX-universal factor graphs
for max-3SAT. However, we present an alternative proof because~\cite{AKKV05} claims the relevant theorem without providing a proof\footnote{Quoting from~\cite{AKKV05}: ``The proof of this theorem, which is a laborious and an almost
exact mimic of the proof of the PCP Theorem, is beyond
the scope of this version of the paper." A subsequent paper~\cite{KPV} that extends~\cite{AKKV05} no longer uses this theorem, and hence does not contain the proof either.}.
%We are not aware of any version of \cite{AKKV05} that contains full proofs.
%The other is that the intended proof in~\cite{AKKV05} (as well as the parts of the proof appearing %in~\cite{KPV}) is patterned after algebraic proofs of the PCP theorem,
Our proof is patterned after a proof of the PCP theorem due to Dinur~\cite{D07}.
%It turns out that a modification of Dinur's proof (technically, we
%modify the writeup of Radhakrishnan and Sudan \cite{RS07}) serves as a transformation from
%polytime-universal factor graphs for 3SAT to APX-universal ones.

Recall that Dinur's proof is based on a sequence of gap amplification steps. However, some of these transformations are not factor graph preserving. Our proof performs a sequence of gap amplifying FGPRs,
starting with the outcome of Theorem \ref{thm:polytime}, and eventually proving Theorem \ref{thm:APX}.
Every FGPR is based on modifying Dinur's proof (or more exactly, on modifying a variation
on Dinur's proof that is given in~\cite{RS07}). The modifications are related to those discussed below for the long code (though our proof for Theorem~\ref{thm:APX} uses a quadratic code rather than the long code).

The proof of Theorem~\ref{thm:threshold1} involves an FGPR from
APX-universal factor graphs for max-3SAT to $77/80$-universal ones.
Our proof is based on a modification of the proof of Bellare et al.~\cite{BGS98}, and consequently obtains the same hardness ratio of $77/80$. The main difficulty we encounter is the following. Tight or nearly tight hardness of approximation results use the so called {\em long code}. A major reason why it is used is that its high redundancy allows one to replace explicit queries that check whether an underlying predicate is satisfied by an implicit operation (referred to as {\em folding}) that allows one to avoid making these queries. The only queries that need to be made are those that check whether the encoding is really (close to) a long code. The saving in queries translates to stronger hardness of approximation results. The problem with folding is that it is sensitive to the predicate that needs to be checked, and a change in the predicate (e.g., changing the polarity of a single variable in a 3SAT clause) changes the folding. As a result, query locations change, and the resulting reduction is not an FGPR. To overcome this problem we introduce a notion of {\em oblivious folding} of the long code, which does allow us to eventually obtain an FGPR. We remark that it was not a-priori obvious that a construct such as oblivious folding should exist at all. In particular, tight hardness of approximation results for 3SAT by Hastad~\cite{H97} use a notion related to folding but somewhat stronger, that is called {\em conditioning} of the long code. We were unable to find an ``oblivious" version of conditioning that can replace the conditioning used by Hastad, and consequently we do not know if $7/8$-universal factor graphs for 3SAT exist.

For the
full proofs of Theorems~\ref{thm:APX} and~\ref{thm:threshold1}, see appendices~\ref{app:APXUFG} and~\ref{sec:TUFG}.

\subsection{Threshold-Universal Families with Nearly Tight Bounds}

%Perhaps unexpectedly, the proof of Theorem~\ref{thm:tight} turns out to follow quite easily from the %statement of Theorem~\ref{thm:threshold}.

Recall that the prefix E (for {\em exact}) in E$k$SAT indicates that every clause in the CNF formula contains exactly $k$ literals (rather than at most) and no two literals in a clause correspond to the same variable. It is not difficult to see that the proof of Theorem~\ref{thm:threshold} in fact gives E3CNF formulas, and not just 3CNF formulas (and even if not, there are simple FPGRs from max-3SAT to max-E3SAT, with only a bounded loss in the approximation ratio). Our proof of Theorem \ref{thm:tight} is based on a direct reduction from instances of max-E3SAT to instances of max-E$k$SAT. This reduction
has the property that mere APX-hardness of max-E3SAT suffices in order
to get nearly tight hardness of approximation ratios for the resulting
max-E$k$SAT instances, if $k$ is sufficiently large.

\begin{proof}
Theorem \ref{thm:threshold} implies that there is a $\left(1-\gamma\right)$-universal
family of factor graphs for E3-CNF formulas, for some $0<\gamma<\nicefrac{1}{8}$.
We shall use this in an FGPR to prove Theorem~\ref{thm:tight}. For simplicity
of the presentation we shall describe our reduction as a reduction
from a single E3-CNF formula $\phi_{3}$ to a single E$k$-CNF formula
$\phi_{k}$. As the factor graph resulting for $\phi_{k}$
will be independent of polarities of variables in $\phi_{3}$, this
will be an FGPR.

Let $\phi_{3}$ be an E3-CNF formula with $n$ variables and $m$
clauses for which one wants to distinguish between the case that it
is satisfiable and the case that it is at most $\left(1-\gamma\right)$-satisfiable.
Formula $\phi_{k}$ will be obtained from a combination of $2^{q}$
auxiliary E$k$-CNF formulas called $\psi_{i}$, for $0\le i\le2^{q}-1$.
Let $q=k-3$. Introduce $q$ fresh variables $y_{1},\ldots,y_{q}$,
and 3 fresh variables $z_{1},z_{2},z_{3}$. Formula $\psi_{0}$ is
obtained from $\phi_{3}$ by adding the $y$ variables (all in negative
polarity) to each clause of $\phi_{3}$. As to the other formulas
indexed by $i\ge1$, each such formula $\psi_{i}$ has eight clauses,
where each clause contains the variables $y_{1},\ldots y_{q},z_{1},z_{2},z_{3}$.
Excluding the all negative polarity combination, there are $2^{q}-1$
remaining combinations of polarities for the $q$ variables of type
$y$. Each such combination of polarities will be associated with
the clauses of one $\psi_{i}$ for $i\ge1$. One may think of the
binary representation of $i$ as specifying the polarity of the $y$
variables in clauses of $\psi_{i}$, where if the $j$'th bit of $i$
is $0$ then $y_{j}$ is negative, and if the $j$'th bit of $i$
is $1$ then $y_{j}$ is positive. As to the $z$ variables, there
are 8 possible combinations of polarities. Within a formula $\psi_{i}$
there are 8 clauses, and each of them has a different combination
of polarities for the $z$ variables.

The formula $\phi_{k}$ will be a weighted mixture of the $\psi_{i}$
(see %remark \ref{Weighted2Unweighted}
Appendix~\ref{sec:tight} regarding an unweighted version).
Formula $\psi_{0}$ is taken with weight $\frac{1}{8\gamma}$ (which
is larger than 1 because $\gamma<\nicefrac{1}{8}$), spreading this
weight equally among its $m$ clauses. Each of the other $\psi_{i}$
is taken with weight 1, spreading the weight equally among its 8 clauses.
The total weight of $\phi_{k}$ is $2^{q}-1+\frac{1}{8\gamma}$.

If $\phi_{3}$ is satisfiable, so is $\phi_{k}$: an assignment to
the original variables of $\phi_{3}$ that satisfies $\phi_{3}$ also
satisfies $\psi_{0}$, and assigning {\em true} to all $y$ variables satisfies
all $\psi_{i}$ for $i\ge1$. If $\phi_{3}$ is only $1-\gamma$ satisfiable
then the weight of unsatisfied clauses in $\phi_{k}$ is at least
$\nicefrac{1}{8}$: if all variables $y$ are assigned {\em true}, this results
from $\psi_{0}$, and in all other cases, this results from one of
the other $\psi_{i}$.

The total weight of $\phi_{k}$ is $W=2^{q}-1+\frac{1}{8\gamma}$,
and for $q$ satisfying $2^{q}\ge\frac{1-\epsilon}{\epsilon}(\frac{1}{8\gamma}-1)$
we have that $W\le\frac{2^{q}}{\left(1-\epsilon\right)}$ which implies
that $\nicefrac{1}{8}\ge\frac{W\left(1-\epsilon\right)}{2^{k}}$.
Hence $\phi_{k}$ is at most $\left(1-\frac{\left(1-\epsilon\right)}{2^{k}}\right)$-satisfiable,
as desired.\end{proof}

\subsection*{Acknowledgements}

Work supported in part by The Israel Science Foundation (grant No. 873/08).

%\bibliographystyle{splncs03}
%\bibliography{us}

\appendix

\section{Polytime-Universal Families}
\label{sec:poly}

Proof of Theorem \ref{thm:polytime}:

\begin{proof}
We design a universal factor graph for 3CNF formulas that have $n$
variables and any number of clauses. For simplicity of presentation,
we use the following convention. A clause is a tuple of three variables
that need not be distinct, and the polarities of the variables. Two
clauses may not have the same tuple, though two clauses may have the
same set of variables if the order in which they appear in the respective
tuples is different. Hence there are exactly $n^{3}$ possible tuples,
and the number of clauses satisfies $m\le n^{3}$.

The universal 3CNF formula is constructed as follows. Write down all
$n^{3}$ possible tuples. For every tuple $T_{i}$ introduce an auxiliary
variable $z_{i}$. For each tuple $T_{i}$ introduce two clauses as
in the following example. If $T_{i}=x_{1},x_{2},x_{3}$ then the two
clauses are $(x_{1}\vee x_{2}\vee z_{i})$ and $(x_{3}\vee z_{i}\vee z_{i})$.
This gives a formula $F$ with $N=n+n^{3}$ variables and $M=2n^{3}$
clauses.

Every 3CNF instance $f$ with $n$ variables can be embedded in $F$,
by appropriately setting only the polarities of variables. For a given
tuple $T_{i}$, if no clause with this tuple is in $f$, then in $F$
give all occurrences of $z_{i}$ positive polarity. By setting $z_{i}$
to true this corresponding tuple drops also from $F$. But if a clause
with tuple $T_{i}$ appears in $f$ (there can be at most one such
clause), then do as in the following example. If the clause is $(x_{1}\vee\bar{x}_{2}\vee x_{3}$)
then in F set the polarities of the clauses derived from $T_{i}$
to $(x_{1}\vee\bar{x}_{2}\vee\bar{z}_{i})$ and $(x_{3}\vee z_{i}\vee z_{i})$.
Any assignment that satisfies these two clauses in $F$ satisfies
also the original clause in $f$.

The above implies that the factor graph of $F$ is polytime-universal
(for 3SAT with $n$ variables). Any algorithm that decides satisfiability
for formulas whose factor graph is that of $F$ can be used to decide
satisfiability of any 3CNF formula with $n$ variables, by following
the above embedding.
\end{proof}

\section{Subexponential-Universal Factor Graphs}
\label{app:SubexpUFG}

In this section we prove Theorem \ref{thm:subexp}. Recall that this involves proving Lemma~\ref{lem:IPZ} and Lemma~\ref{lem:polylog}, and then combining them appropriately.

\begin{lemma}
\label{lem:IPZ} Given a 3CNF formula $\varphi$ with
$n$ variables and any number of clauses (at most $O(n^3)$ as clauses may be assumed to be distinct), for every $0 < \epsilon < 1/10$
there is an algorithm that runs in time $2^{O(n^{1 - \epsilon})}$ and produces at most $2^{n^{1 - \epsilon}}$
new 3CNF formulas, each with $n$ variables and at most $O(n^{1 + 2\epsilon} (\log n)^2)$
clauses,
such that $\varphi$ is satisfiable iff at least one of the new
formulas is satisfiable.
\end{lemma}

Lemma~\ref{lem:IPZ} follows by substituting $\epsilon(n) = n^{-\epsilon}$ in the following lemma.

\begin{lemma}
\label{LemmaLinearSize}Given a $k$-CNF formula, $\varphi$, with
$n$ variables, $0<\epsilon\left(n\right)<1$, and $\alpha\left(n\right)$
with $\frac{\alpha\left(n\right)}{\log4\alpha\left(n\right)}>4k2^{k-1}\epsilon^{-1}\left(n\right)$,
there is an algorithm that produces at most $2^{\epsilon\left(n\right)n}$
$k$-CNF formulas, each with at most $nk\left(4\alpha\left(n\right)\right)^{2^{k-2}}$
clauses and $n$ variables in time $2^{\epsilon\left(n\right)n}n^{\left(4\alpha\left(n\right)\right)^{2^{k-2}}}\mathrm{poly}\left(n\right)$
such that $\varphi$ is satisfiable iff at least one of the outputted
formulas is satisfiable.
\end{lemma}

\begin{proof} A similar statement was proved by Impagliazzo et al.~\cite{IPZ01}
with constant $\epsilon$ and $\alpha$, and the same proof works
when they are not constant.\end{proof}

We now prove Lemma~\ref{lem:polylog}.

\begin{proof}
We first construct a nondeterministic circuit that receives as input an E3CNF formula with $n$ variables and $m$ clauses and outputs~1 if the formula is satisfiable. We wish to keep the circuit small, of size $\mathrm{O}\left(m\log m\log n\right)$. For this reason, the nondeterministic aspect of the circuit will not be a guess of the assignment to the variables (which amounts to $n$ nondeterministic guesses), but rather a selection of one index per clause (hence $m$ nondetermistic guesses, each among three possibilities), indicating a literal that satisfies this clause. The consistency of all these selections (namely, not selecting a variable in one clause and its negation in a different clause) will be checked using a circuit that mimics an oblivious sorting network. All selected literals will be sorted, implying that if there is a variable who was selected both positively and negatively, these two contradicting selections will ``meet" during the sorting processes and the inconsistency will be detected. As there are oblivious sorting networks that sort $m$ numbers using $O(m \log m)$ comparisons~\cite{AKS83}, the size of the circuit will remain bounded by $O(m \log m \log n)$ (the extra $\log n$ term comes from the fact that it takes $\log n$ bits to specify each of the sorted numbers).
Such a nondeterministic circuit outputs 1 iff the
formula is satisfiable: a consistent selection of literals can always be completed to a satisfying assignment (by giving arbitrary values to variables for which no occurrence of their literals was selected), whereas given a satisfying assignment a consistent selection is obtained by selecting the first satisfied literal in every clause.

We now provide more details on the construction of the circuit.
The circuit takes
$\mathrm{O}\left(3m\log n\right)$ input bits. $\mathrm{O}\left(\log n\right)$
bits are used to represent each literal, with the least significant
bit used to indicate if the literal is negated or not. In addition,
there are 2 nondeterministic input bits per clause, used to select one of the three literals
in the clause. The selection of
literals can be done by using $\mathrm{O}\left(\log n\right)$ 3-to-1
multiplexers. The selected literals are sorted using a sorting network
of size $\mathrm{O}\left(m\log m\right)$ (see \cite{AKS83,P90}),
where each comparison is done by adding the representation of one
literal to the two's complement of the representation of the other
literal, using $\mathrm{O}\left(\log n\right)$ adders, and the most
significant bit of the result determines the output of the comparison.
The literals are switched or not, depending on the result of the comparison,
using $\mathrm{O}\left(\log n\right)$ 2-to-1 multiplexers. Lastly,
when all literals are sorted, each consecutive pair (with overlapping
pairs) is checked that it does not contain the representation of a
variable and its negation (all but last bit equal, using $\mathrm{O}\left(\log n\right)$
gates).

Given an E3CNF formula $\varphi$ with $m$ clauses and $n$ variables,
we use the circuit described above, to construct a 3-CNF formula $\Phi_{\varphi}$
that is satisfiable iff $\varphi$ is satisfiable. Additionally, if
$\psi$ is another E3CNF formula with $m$ clauses and $n$ variables,
the factor graph of $\Phi_{\psi}$ is also the factor graph of $\Phi_{\varphi}$.\\
For every input of a gate and for the output of the circuit there
will be a variable. Note that the output of every gate is an input
of some other gate or the output of the circuit.\\
Each gate contributes a bounded number of clauses to $\Phi_{\varphi}$ that encode the requirement that
the output of the gate is correct. For example, a NAND gate with inputs $x,y$
and output $z$ would contribute the clauses $\bar{x}\vee\bar{y}\vee\bar{z}$,
$\bar{x}\vee y\vee z$, $x\vee y\vee z$, $x\vee\bar{y}\vee z$. This ensures that a satisfying assignment to
$\Phi_{\varphi}$ is a
valid calculation of the circuit. That is, every variable has the
value passed to or from each gate.\\
Let $o$ be the variable representing the output of the circuit.
The clause $o$ is also added to $\Phi_{\varphi}$. This ensures that
an assignment satisfies the formula iff the output of the circuit
is 1.\\
Let $i_{1},\cdots i_{3m\left(\left\lceil \log_{2}n\right\rceil +1\right)}$
be the input bits of the circuit. For every $j$, either the clause
$i_{j}$ or the clause $\bar{i_{j}}$ is added to $\Phi_{\varphi}$,
depending on the representation of $\varphi$. This ensures that a
satisfiable assignment to $\Phi_{\varphi}$ has the representation
of $\varphi$ in the input of the circuit. Note that the only difference
between $\Phi_{\varphi}$ and $\Phi_{\psi}$ is in the polarity of
these clauses.

A standard transformation can be used to transform the formula from
3CNF to E3CNF.

If $\Phi_{\varphi}$ is satisfiable, the variables representing the
selector bits prove that the circuit can be made to output 1, when
$\varphi$ is given as input. If there are selector bits that make
the circuit output 1 on $\varphi$, setting each variable to its respective
input/output in the circuit, shows that $\Phi_{\varphi}$ is satisfiable.
\end{proof}

Equipped with Lemmas~\ref{lem:IPZ} and~\ref{lem:polylog}, we now prove Theorem \ref{thm:subexp}.

\begin{proof}
For simplicity of the presentation, we omit the $O$ notation in the expressions that we derive.

Assume that for some $0 < \delta < 1/10$ a hypothetical algorithm H can solve any instance on the universal factor graphs of Lemma~\ref{lem:polylog} in time $2^{N^{1 - \delta}}$. Consider now an arbitrary 3SAT instance with $n$ variables. For $\epsilon = \delta/3$, use Lemma~\ref{lem:IPZ} to create $2^{n^{1 - \epsilon}}$ new 3CNF formulas with at most $n^{1+2\epsilon}(\log n)^2$ clauses. Use Lemma~\ref{lem:polylog} to reduce every such 3CNF instance to an instance on a universal factor graph with $N = n^{1+2\epsilon}(\log n)^4$ variables. Use algorithm $H$ to solve these instances, thus obtaining the solution to the original 3SAT instance. The choice of $\epsilon = \delta/3$ implies that this whole procedure takes time roughly $2^{n^{1 - \epsilon}}$.
\end{proof}

\section{APX-Universal Factor Graphs\label{app:APXUFG}}

In this section we prove Theorem~\ref{thm:APX}.

Any universal factor graph for 3SAT (e.g., the result of Theorem~\ref{thm:polytime}) is $\left(1-\frac{1}{m}\right)$-universal, where $m$ is the number of clauses.
In order
to create a $(\mbox{1-\ensuremath{\epsilon}})$-universal factor graph
(for some $\epsilon>0$) the instances will go through an iterative process,
increasing the worst case unsatisfiability of the formulas by a factor
of 2, while increasing the size of the factor graph by a constant
factor. The construction is based on the combinatorial method to prove
the PCP theorem by Dinur \cite{D07}, and closely follows the proof
of Radhakrishnan and Sudan \cite{RS07}. Familiarity with these earlier proofs (an overview of which can be found in~\cite{RS07}) can aide the reader in following our proof.

\subsection{Definitions}

\begin{definition}A \emph{constraint satisfaction problem} (CSP) has the
form $P= \left(V,\Sigma,C\right)$, where $v$ is the set of variables,
$\Sigma$ is the alphabet, and $C$ is the set of constraints. A constraint
is $c=\left\langle U,f\right\rangle $, where $U\subset V$ and $f:\Sigma^{U}\to\left\{ 0,1\right\} $.
An \emph{assignment} is a function $a:V\to\Sigma$, giving each variable
a value. Given an assignment $a$, a constraint $c=\left\langle U,f\right\rangle $
is said to be \emph{satisfied} by the assignment (usually, the assignment
will be implied from the context) if $f\left(a|_{U}\right)=1$, otherwise
it is \emph{unsatisfied} by the assignment.

Given a constraint satisfaction problem $P$, $\mathrm{UNSAT}\left(P\right)$
is the minimal fraction of constraints (over all assignments) that
are unsatisfied. The \emph{size} of $P$ is $\left|P\right|=\left|V\right|+\left|C\right|$.\end{definition}

\begin{definition}

A \emph{constraint hypergraph} $H=\left(V,E,\Sigma,C\right)$ is an
alternative definition of a CSP, where $V,\Sigma,C$ are as in the
definition of a CSP, and for every constraint $c=\left\langle U,f\right\rangle $,
$U\in E$.

The \emph{structure} of a constraint hypergraph $H=\left(V,E,\Sigma,C\right)$
is the hypergraph $\left(V,E\right)$.

The \emph{rank} of $H$ is the maximal cardinality of an edge.

In the special case where all sets in $E$ have cardinality 2, $H$
is a \emph{constraint graph}.\end{definition}

The following definition adds parametrization to the earlier definition of FGPR given in Section~\ref{sec:threshold}.

\begin{definition}
A $\left(\delta,d\right)$\emph{-FGPR} (Factor Graph Preserving Reduction)
is a transformation of instances of one class of CSP to instances
of another class of CSP with the additional requirements:\end{definition}
\begin{itemize}
\item If the factor graphs of $A$ and $B$ are equal, then the factor graphs
of their transformations are equal.

\item There is some constant $\xi>0$ such that if $A$ is transformed to
$A'$ then:

\begin{itemize}
\item if $\mathrm{UNSAT}\left(A\right)=0$, then $\mathrm{UNSAT}\left(A'\right)=0$.
\item if $\mathrm{UNSAT}\left(A\right)\geq\epsilon$, then $\mathrm{UNSAT}\left(A'\right)\geq\delta\min\left\{ \epsilon,\xi\right\} $.
\end{itemize}
\item $d\left|A\right|\geq\left|A'\right|$
\end{itemize}
For example, the standard reduction from 3SAT to E3SAT is a $\left(\nicefrac{1}{4},4\right)$-FGPR ($\left(\nicefrac{1}{2},2\right)$-FGPR
if all clauses have at least two distinct variables).

In order to create an APX-universal factor graph a $\left(\delta,d\right)$-FGPR
with $\delta>1$ will be constructed, using a composition of several
FGPRs. In order to compose these FGPRs correctly, each will need to
have additional properties.

It will convenient for us to represent constraints as polynomials. For example, a constraint of the form
{}``the first bit in the representation of the variable $x$ is equal
to the second bit of the representation of the variable $y$'' can
be represented as requiring that polynomial $x_{1}+y_{2}$ be equal
to 0 (where $x_{1}$ refers to the value of the first bit in the assignment
of $x$, and $y_{2}$ refers to the second bit of the assignment of
$y$).
\begin{definition}
Let $\Sigma=\mathbb{F}_{2}^{k}$. A constraint $e$ is \emph{m-restricted
}if it can be represented as a set of up to $m$ polynomials of degree
two, $\left\{ P_{i}^{e}\right\} $, such that the constraint is satisfied
iff all the polynomials are 0 (where the assignment of variables is
treated as the values of $k$ Boolean variables). In such case we
say that $e$ is \emph{associated} with $\left\{ P_{i}^{e}\right\} $.
An \emph{m-restricted constraint (hyper)graph} is a constraint (hyper)graph
with only $m$-restricted constraints.
\begin{definition}
Two $m$-restricted constraint hypergraphs are \emph{close} if they
share the same factor graph and alphabet, and for every edge $e$
of the factor graph, if the constraint corresponding to $e$ is associated
with $\left\{ P_{i}^{e}\right\} $ then the constraint corresponding
to $e$ in the other graph is associated with $\left\{ P_{i}^{e}+b_{i}^{e}\right\} $,
where $b_{i}^{e}\in\left\{ 0,1\right\} $.
\end{definition}
\end{definition}

\subsection{Changing the Representation}
\begin{lemma}
\label{Formula2Graph}There is an explicit $\left(\nicefrac{1}{4},6\right)$-FGPR
from 3-SAT to 2-restricted constraint graphs. Additionally, if $H,H'$
are created from $\varphi,\varphi'$, respectively, using the specified
reduction and $\varphi,\varphi'$ have the same factor graph then
$H$ and $H'$ are close.\end{lemma}
\begin{proof}
Each vertex in the constraint graph will represent up to two literals
plus an additional bit, so $\Sigma=\mathbb{F}_{2}^{3}$. For a vertex
$w$, $A_{i}\left(w\right)$ will correspond to the $i$'th bit of
the assignment of the vertex. The value of the first bits is intended
to be the value of the represented literals, 0 if true, 1 if false
(and the constraints will try to enforce that).

For every variable $v_{i}$ in the formula there will be a corresponding
vertex $w_{i}$. For every clause of the form $x_{i}\vee x_{j}\vee x_{k}$
(where $x_{\ell}$ is $v_{\ell}$ or $\bar{v}_{\ell}$) there will
be two vertices: $u_{ij},u_{k}$ with an edge between them. The constraint
corresponding to the edge $\left(u_{ij},u_{k}\right)$ expects the
following two polynomials to be satisfied: $A_{3}\left(u_{ij}\right)=A_{1}\left(u_{ij}\right)A_{2}\left(u_{ij}\right)$,
(the last bit is true if one of the represented literals is true)
and $A_{3}\left(u_{ij}\right)A_{1}\left(u_{k}\right)=0$ (the clause
is satisfied). For every clause of the form $x_{i}\vee x_{j}$ there
will be two vertices, $u_{i}$ and $u_{j}$ with an edge between them
with the constraint $A\left(u_{i}\right)A\left(u_{j}\right)=0$ (the
clause is satisfied). For every clause of the form $x_{i}$ a self
loop will be added to the vertex $w_{i}$ with constraint $A_{1}\left(v_{i}\right)=b$,
where $b=0$ if $x_{\ell}$ is $v_{\ell}$ and $b=1$ otherwise.

In addition, consistency constraints will be added: $\left(u_{ij},w_{i}\right),\left(u_{ij},w_{j}\right),\left(u_{k},w_{k}\right)$
with respective constraints $A_{1}\left(u_{ij}\right)-A_{1}\left(w_{i}\right)=b_{i}^{e}$,
$A_{2}\left(u_{ij}\right)-A_{1}\left(w_{j}\right)=b_{j}^{e}$, $A_{1}\left(u_{k}\right)-A_{1}\left(w_{k}\right)=b_{k}^{e}$,
where $b_{\ell}^{e}$ is 0 or 1, depending on whether the variable
$v_{\ell}$ or its negation appear in the clause the first vertex
of the edge corresponds to.

Every clause is responsible for the creation of at most two vertices
and four edges. Every variable is responsible for the creation of
one vertex. Thus $\left|H\right|\leq6\left|\varphi\right|$.

If $\mathrm{UNSAT}\left(\varphi\right)=0$, there is a assignment
$a_{i}$ for each $v_{i}$ satisfying all clauses. setting $A_{1}\left(w_{\ell}\right)$
according to this assignment (0 if $a_{\ell}$ is true, 1 otherwise),
and setting $A_{1}\left(u_{k}\right),A_{1}\left(u_{ij}\right),A_{2}\left(u_{ij}\right)$
according to the value of the corresponding literal using the assignment
will satisfy all edges of the graph.

If $\mathrm{UNSAT}\left(G\right)<\epsilon$, the best assignment to
vertices can be transformed into an assignment to variables. If an
edge is unsatisfied, the clause that generated that edge is considered
to be unsatisfied. All edges and variables generated by this clause
will be removed. Repeating this for as long as unsatisfied edges remain,
we are left with a completely satisfiable graph, with all consistency
constraints holding. Thus, the assignment for the graph can be transformed
to an assignment for the formula. Each unsatisfied edge may have caused
a single clause to be unsatisfied, and since there are at most 4 times
as many edges as there are clauses, $\mathrm{UNSAT}\left(\varphi\right)<4\epsilon$.

Lastly, it is immediate that all formulas that have the same factor
graph generate close constraint graphs.
\end{proof}

\subsection{Gap Amplification}
\begin{definition}
An $\left(\eta,d\right)$\emph{-expander} is a regular graph $G=\left(V,E\right)$
with degree $d$ and for all $S\subset V$ with $\left|S\right|\leq\nicefrac{\left|V\right|}{2}$,
$\left|\left\{ \left(u,v\right)\in E|u\in S,v\notin S\right\} \right|\geq\eta\left|S\right|$.
An \emph{$\left(\eta,d\right)$}-expander is \emph{positive} if every
vertex has at least $\nicefrac{d}{2}$ self loops.\end{definition}
\begin{theorem}
[See several constructions in \cite{HLW06}]There are $\eta,d>0$
such that positive $\left(\eta,d\right)$-expander graphs exist on
$n$ vertices, for all $n>0$.
\end{theorem}

\begin{theorem}
\label{GapAmplification}There is a universal constant $\alpha$ such
that for every $k,m,t\in\mathbb{N}$, there is an explicit $\left(\delta_{1}^{*}\left(t\right),c_{1}\left(t\right)\right)$-FGPR
from (m-restricted) constraint graphs with alphabet \textup{$\mathbb{F}_{2}^{k}$}
to ($\mathrm{O}\left(mt\right)$-restricted) constraint graphs with
alphabet \textup{$\mathbb{F}_{2}^{s\left(k,t\right)}$ }with \textup{$\delta_{1}^{*}\left(t\right)\geq\alpha t$.}

Furthermore, every constraint of the produced constraint graph is
a conjunction of $\mathrm{O}\left(t\right)$ constraints from the
input graph and equality constraints, and the set of constraints only
depends on the factor graph of the input.

Specifically, two m-restricted close constraint graph are transformed
into two close $\mathrm{O}\left(mt\right)$-restricted constraint
graphs.
\end{theorem}

\begin{proof}
The proof closely follows the construction of the transformation in
Section 5 in \cite{RS07}.

Let the input of the transformation be a graph $G$. As in Lemma 5.3
in \cite{RS07}, a regular graph $G_{1}$ is created from the graph
$G$. Each vertex $u\in V$ is replaced by $\mathrm{d}_{u}$ ($u$'s
degree) vertices, with an $\left(\eta,d\right)$-expander embedded
on them. In addition, each of the new vertices is connected to one
of $u$'s neighbors. The constraints on the expanders' edges are equality
constraints (polynomials of degree 1). $\mathrm{UNSAT}\left(G_{1}\right)\geq\delta_{1}\mathrm{UNSAT}\left(G\right)$,
where $\delta_{1}$ depends on $d$ and $\eta$ (which are constants).
The proof for the last claim is contained in \cite{RS07}. Also, If
$G$ is satisfied it is immediate that $G_{1}$ can be satisfied.

As in Lemma 5.5 in \cite{RS07}, an expander $G_{2}$ is created from
the graph $G_{1}$ by superimposing an $\left(\eta,d\right)$-expander
on the graph $G_{1}$, with constraints that are always satisfied.
Let $d_{0}$ be the degree of $G_{1}$. Then, immediately, $G_{2}$
is a $\left(\eta,d+d_{0}\right)$ expander. By adding $d+d_{0}$ self
loops with constraints that are always satisfied (polynomials of degree
0), $G_{2}$ is a positive $\left(\eta,2d+2d_{0}\right)$-expander.
$\mathrm{UNSAT}\left(G_{2}\right)\geq\delta_{2}\mathrm{UNSAT}\left(G_{1}\right)$,
where $\delta_{2}$ depends on $d$ (if $\epsilon$ of the constraints
are not satisfied in $G_{2}$, then $\delta_{2}\epsilon$ of the constraints
are not satisfied in $G_{1}$, using the same assignment). Also, an
assignment that satisfies $G_{1}$ satisfies $G_{2}$.

Lastly, we transform $G_{2}=\left(V_{2},E_{2},\Sigma_{2},C_{2}\right)$
to $G_{3}=\left(V_{2},E,\Sigma_{2}^{\left(d+d_{0}\right)^{t}},C_{3}\right)$,
where $G_{3}$ is the product constraint graph of $G_{2}$, as defined
in definition 5.13 in \cite{RS07}. The value of $v\in V_{2}$ is
supposed to be the concatenation of all the values of vertices of
distance at most $t$ from $v$ in $G_{2}$ ($\left(d+d_{0}\right)^{t}$
is an upper bound on the number of vertices at distance at most $t$
from any vertex). That is, for an assignment $A$ and every two vertices
$v,u\in V_{2}$ there is $A_{u}\left(v\right)\in\Sigma_{2}\cup\left\{ \emptyset\right\} $
(where $\emptyset\notin\Sigma_{2}$), the opinion of $v$ on $u$.
$A_{u}\left(v\right)=\emptyset$ iff the distance between $u$ and
$v$ is more than $t$, and then we say that $v$ has no opinion on
$u$.

In order to define the edges, we use some arbitrary order on the $d$
neighbors of each vertex. The edges are intended to represent a simple
random walk on the graph starting at a random vertex and stopping
after each step with probability $\nicefrac{1}{t}$. If the random
walk does not stop after $5t+1$ steps we call it a null walk and
terminate it. There is one edge for each element of the set $V\times\left(\left\{ 1\ldots d\right\} \times\left\{ 1\ldots t\right\} \right)^{5t}$.
Each edge and its corresponding constraint are determined by a walk
defined by $\left\langle a,\left\langle i_{1},j_{1}\right\rangle \ldots\left\langle i_{5t},j_{5t}\right\rangle \right\rangle $
in the following way: The walk starts from $a=v_{0}$ (which is in
$V_{2}$). In step $k$ (starting from $k=1$), the walk moves to
the neighbor of $v_{k-1}$ numbered by $i_{k}$, and calls this vertex
$v_{k}$. If $j_{k}=1$, we stop (to simulate stopping with probability
$\nicefrac{1}{t}$), otherwise, we continue to step $k+1$, until
$k=5t$. We call the last vertex reached $b$. If the walk stops before
reaching step $5t+1$, a walk that is not null, the corresponding
constraint is checking that the opinions of $a$ and $b$ are the
same on the vertices on the path (including $a$ and $b$), if both
have opinions on the vertices and that the constraints of $G_{2}$
on the edges of the path are satisfied. Note that equality constraints
can be modeled as requiring that a linear polynomial is 0. Thus, this
creates an $\mathrm{O}\left(mt\right)$-restricted constraint. For
a null walk, the constraint is a self loop that is always satisfied.

Note that the transformations transformed close graphs to close graphs
and that the size of the graph increased by a factor that depends
only on $t$ (since degrees of expanders were chosen to be constants)

The proof for the last transformation increasing the unsatisfiability
is described throughout most of Section 5 of \cite{RS07}.
\end{proof}

\subsection{Alphabet Reduction}
\begin{definition}
The \emph{Hadamard code} of a binary string $x=x_{1}x_{2}\ldots x_{\ell}\in\left\{ 0,1\right\} ^{\ell}$
is the value of all linear functions $\left\{ 0,1\right\} ^{\ell}\to\left\{ 0,1\right\} $
on the bits of $x$. Given some arbitrary ordering on the linear functions,
the $i$'th bit of the Hadamard code of $x$ is the value of $x$
on the $i$'th function.

The \emph{quadratic code} of a binary string $x=x_{1}x_{2}\ldots x_{\ell}\in\left\{ 0,1\right\} ^{\ell}$
is the value of all homogeneous quadratic functions $\left\{ 0,1\right\} ^{\ell}\to\left\{ 0,1\right\} $
on the bits of $x$. Given some arbitrary ordering on the quadratic
functions, the $i$'th bit of the quadratic code of $x$ is the value
of $x$ on the $i$'th function. \end{definition}
\begin{lemma}
\label{AlphabetReduction}For every $k,m\in\mathbb{N}$ there is an
explicit $\left(\delta_{2}^{*},c_{2}\left(k,m\right)\right)$-FGPR
from close $m$-restricted constraint graphs with alphabet $\mathbb{F}_{2}^{k}$
to close 1-restricted constraint hypergraphs with alphabet $\mathbb{F}_{2}$.
Additionally, the constructed hypergraph has rank 4.\end{lemma}
\begin{proof}
Given an $m$-restricted constraint graph $G=\left(V,E,\mathbb{F}_{2}^{k},C\right)$,
we construct a 1-restricted constraint hypergraph $H=\left(V',E',\mathbb{F}_{2},C'\right)$.
Following the construction of Section 6.3 in \cite{RS07}, for every
$v\in V$, and linear function $L:\left\{ 0,1\right\} ^{k}\to\left\{ 0,1\right\} $
create a vertex $v\left(L\right)$ in $V'$. For every $e\in E$,
for every homogeneous quadratic function $q:\left\{ 0,1\right\} ^{2k}\to\left\{ 0,1\right\} $
create a vertex $e\left(q\right)$ in $V'$. For every $e\in E$,
for every linear function $L:\left\{ 0,1\right\} ^{2k}\to\left\{ 0,1\right\} $
create a vertex $e\left(L\right)$ in $V'$. The value of the assignment
$A$ on vertex $v$ is referred to as $A\left(v\right)$.

For every $e=\left(u,v\right)\in E$, $L_{1},L_{2}:\left\{ 0,1\right\} ^{k}\to\left\{ 0,1\right\} $,
$L_{3},L_{4}:\left\{ 0,1\right\} ^{2k}\to\left\{ 0,1\right\} $ linear
functions, $q_{1},q_{2}:\left\{ 0,1\right\} ^{2k}\to\left\{ 0,1\right\} $
homogeneous quadratic functions, $w\in\left\{ 0,1\right\} ^{m}$ there
will be seven constraints and corresponding edges (with multiplicity):

\begin{enumerate}

\item A constraint that is satisfied iff $A\left(u\left(L_{1}\right)\right)+A\left(u\left(L_{2}\right)\right)=A\left(u\left(L_{1}+L_{2}\right)\right)$.

\item A constraint that is satisfied iff $A\left(v\left(L_{1}\right)\right)+A\left(v\left(L_{2}\right)\right)=A\left(v\left(L_{1}+L_{2}\right)\right)$.

\item A constraint that is satisfied iff $A\left(e\left(L_{3}\right)\right)+A\left(e\left(L_{4}\right)\right)=A\left(e\left(L_{3}+L_{4}\right)\right)$.

\item A constraint that is satisfied iff $A\left(e\left(q_{1}\right)\right)+A\left(e\left(q_{2}\right)\right)=A\left(e\left(q_{1}+q_{2}\right)\right)$.

\item Let $L$ be the linear function given by $L\left(X,Y\right)=L_{1}\left(X\right)+L_{2}\left(Y\right)$.
There is a constraint that is satisfied iff $A\left(u\left(L_{1}\right)\right)+A\left(v\left(L_{2}\right)\right)=A\left(e\left(L\right)\right)$.

\item A constraint that is satisfied iff $A\left(e\left(L_{3}\right)\right)A\left(e\left(L_{4}\right)\right)=A\left(e\left(q_{1}+L_{3}L_{4}\right)\right)-A\left(e\left(q_{1}\right)\right)$.

\item Let $P$ be a homogeneous degree two polynomial and $b\in\left\{ 0,1\right\} $
such that $P=\Sigma w_{i}P_{i}+b$, where $\left\{ P_{i}\right\} $
is the set of polynomials associated with $e$. There is a constraint
that is satisfied iff $A\left(e\left(q_{1}+P\right)\right)-A\left(q_{1}\right)-b=0$.

\end{enumerate}

There are $2^{\ell}$ linear functions on $\ell$ bits. There are
$2^{\frac{\ell\left(\ell-1\right)}{2}}$ quadratic functions on $\ell$
bits. Thus, there are $2^{k}\left|V\right|+\left(2^{2k}+2^{k\left(2k-1\right)}\right)\left|E\right|$
vertices and $7\left|E\right|\cdot2^{m+6k+2k\left(2k-1\right)}$ edges.
Hence $\left|H\right|\leq2^{O\left(k^{2}+m\right)}\left|G\right|$.

The satisfaction of every hyperedge depends only on the value of at
most 4 vertices.

Given an assignment $A$ satisfying $G$, there is an assignment $A'$
satisfying $H$. For every $i,v$, assign the $i$'th bit of the Hadamard
code of $A\left(v\right)$ to the vertex $v\left(L_{i}\right)$. Note
that this satisfies all constraints of type 1 and 2, since the Hadamard
code is linear. For every $i,e=\left(u,v\right)$ assign the $i$'th
bit of the Hadamard code of $A\left(u\right)\circ A\left(v\right)$
(the concatenation of the binary strings) to the vertex $e\left(L_{i}\right)$.
Note that this satisfies all constraints of type 3. Assign the $i$'th
bit of the quadratic code of $A\left(u\right)\circ A\left(v\right)$
to the vertex $e\left(q_{i}\right)$. A quadratic code of any string
$x$ on the vertices $e\left(q\right)$ (for all quadratic functions
$q$), will pass the constraints of type 6, if the Hadamard code of
$x$ is on the vertices of $e\left(L\right)$ (for all linear functions
$L$), and this is the case for $A'$. The constraints of type 5 check
the consistency of the Hadamard code between vertices of $e\left(L\right)$
and vertices corresponding to $u$ and $v$, and they are consistent
in $A'$. Lastly, constraints of type 7 are satisfied, since the polynomials
associated with edges of $G$ are all 0, when given assignment $A$.
Since the code on $e\left(q\right)$ is linear, the check becomes
$A'\left(e\left(P\right)\right)-b$, which, by definition of $A'$,
is the value of $P-b$ on the assignment $A$, which is a sum of polynomials
that are all 0.

If two constraint graphs $G$ and $G^{*}$ are close, then their transformation,$H$
and $H^{*}$, only differ in the constraints of type 7, in the constant
of the associated polynomials (actually linear functions).

The proof that the transformation decreases the fraction of unsatisfiable
constraints by at most a constant factor is the same as in Lemma 6.11
in \cite{RS07}.
\end{proof}

\subsection{Composition}

Firstly, we transform the hypergraph back to a 3CNF formula.
\begin{lemma}
\label{Hypergraph2Formula}There is an explicit $\left(\delta_{3}^{*}\left(h\right),c_{3}\left(h\right)\right)$-FGPR
from close 1-restricted constraint hypergraphs with alphabet $\mathbb{F}_{2}$
to 3-SAT, where $h$ the rank of the hypergraph.\end{lemma}
\begin{proof}
Given a 1-restricted constraint hypergraph $H=\left(V,E,\Sigma,C\right)$,
a 3-CNF formula $\varphi$ is defined.

For every vertex $v\in V$, there is a variable $v$. For every edge
$e$, there is a variable $w_{e}$.

For every edge $e$, the corresponding constraint $c_{e}$ is satisfied
iff the quadratic polynomial $P^{e}+b^{e}$ is evaluated to 0, where
$P^{e}$ is homogeneous. There is a set of at most $2^{h}$ 3-CNF
clauses that is evaluated to true iff $P^{e}\left(u_{1},u_{2},\cdots,u_{h}\right)+w_{e}=0$.
Adding to this set the clause $w_{e}$ if $b_{i}=1$ and the clause
$\bar{w_{e}}$ otherwise, gives a set of at most $2^{h}+1$ clauses
that are all satisfied iff the constraint $c_{e}$ is satisfied (using
the same assignment, omitting the variables of the form $w_{e}$).

Thus, we have that $\left|\varphi\right|\leq\left(2^{h}+1\right)\left|H\right|$.
Also, if $\mathrm{\mathrm{UNSAT}}\left(\varphi\right)\leq\epsilon$,
then $\mathrm{\mathrm{UNSAT}}\left(H\right)\leq\left(2^{h}+1\right)\epsilon$.
Finally, it is immediate that transforming close 1-restricted constraint
hypergraph the resulting 3CNF formulas all have the same factor graph.
\end{proof}
Now we can compose all transformation to get the required FGPR.
\begin{theorem}
\label{DoublingGap}There is an explicit $\left(2,c_{4}\right)$-FGPR
from 3-SAT to 3-SAT.\end{theorem}
\begin{proof}
Using lemma \ref{Formula2Graph}, theorem \ref{GapAmplification},
lemma \ref{AlphabetReduction}, and lemma \ref{Hypergraph2Formula},
there is a \[
\left(\frac{1}{4}\alpha t\delta_{2}^{*}\delta_{3}^{*}\left(4\right),6c_{1}\left(t\right)c_{2}\left(2^{s\left(3,t\right)},\mathrm{O}\left(t\right)\right)c_{3}\left(4\right)\right)-FGPR\]
 from 3-SAT to 3-SAT formulas, for all $t$ (it is easy to verify
that the composition of these FGPRs is indeed an FGPR). Specifically,
there is a constant $t$ to get a $\left(2,c_{4}\left(t\right)\right)$-FGPR.
\end{proof}
Starting with a universal factor graph for 3-SAT, $\mathrm{O}\left(\log n\right)$
repetitions of Theorem \ref{DoublingGap} proves Theorem \ref{thm:APX}.

\section{Threshold-Universal Factor Graphs}\label{sec:TUFG}

\subsection{Oblivious folding of the long code - an overview}

A major ingredient in tight hardness of approximation results is the
\emph{long code}, introduced in \cite{BGS98}. We present here an overview of the difficulties involved in adapting hardness proofs based on the long code to the oblivious setting, and of our approach for handling these difficulties. A superficial familiarity with previous work should suffice in order to follow this overview -- there is no need to know concepts related to the analysis of the long code, such as dictatorship tests and Fourier analysis.

For a given value of
$k$, the long code replaces a vector $x$ of $k$ original Boolean variables
by a vector $z$ of $2^{2^{k}}$ new Boolean variables. The method of doing this can be visualized as follows.
Consider a $2^{k}$ by $2^{2^{k}}$ matrix LC (for \emph{long code}).
The rows are indexed by all $2^{k}$ possible assignments to the
$x$ variables. The columns are indexed by all possible Boolean functions.
Namely, the column vectors are all possible $2^{2^{k}}$ truth-tables
for Boolean functions on $k$ variables (equivalently, all possible
column vectors of dimension $2^{k}$). Every row of the matrix is
then the \emph{long code} of its corresponding assignment. The $z$
variables are intended to correspond to the columns, and their values
(as a vector) are intended to be equal to one row in the matrix (one
long code). A \emph{verifier} may perform various tests to verify
(in a probabilistic sense) that this is indeed the case. Moreover,
the verifier would like this row to correspond to an assignment to
the $x$ variables that satisfies some predicate $h$. The concept
of \emph{folding} assists in these tests.

\emph{Folding over }\textbf{\emph{true}}. Columns in LC can be paired,
where each column (say $z_i$) is paired with the column (say $z_j$) that is its complement.
If $z$ is indeed a long code, then in each such pair, one of the two corresponding variables is redundant and hence is dropped.
For example, to read the value of $z_{j}$, read the value of $z_{i}$
and flip this value. After folding over true, half the number of $z$ variables remain.
%This operation of dropping half the $z$ variables is referred to as \emph{folding over true}.

\emph{Folding over $h$.} Consider the truth table for $h$ (as a
column vector). Only certain rows of LC have a value of~1 in this
vector, and these rows correspond to the assignments that satisfy
the predicate $h$. In these rows, the following holds. Consider an
arbitrary variable $z_{i}$ whose column corresponds to the function
$f$. Then its value is exactly the complement of the variable $z_{j}$
whose column corresponds to the function $f+h$ (addition modulo~2).
Hence again columns can be paired (columns that differ by $h$), and
one member from each pair can be dropped. After folding over $h$, half the number of $z$ variables remain. We remark that folding over true is a special case of folding over $h$, for a trivial choice of $h$ as the predicate that always accepts (its truth table is all~1).
%This operation is referred to as \emph{folding over $h$}.

\emph{Conditioning over $h$.} Consider the rows correspond to the
assignments that satisfy the predicate $h$, and let $r$ denote their number. The columns restricted
to these rows can be partitioned into $2^r$ equivalence classes, based on equality.
%For example, if there are~4 rows satisfied by $h$, then all columns
%of the form $(0,1,1,0)^{T}$ will form one equivalent class.
Only one member from each equivalence class needs to be retained, as its
$z$-variable determines the value of the corresponding $z$ variables
for all other columns in its equivalence class. Hence after conditioning, the number of $z$ variables that would remain is $2^r$.
%This operation, introduced in \cite{H97}, is called \emph{conditioning over $h$}.

Hardness of approximation results involve {\em long code tests} that query some of the $z$ variables and determine whether their values are consistent with the assumption that all $z$ variables form a valid long code. Basically, there are two versions of long code tests. One version (as in Bellare et al \cite{BGS98}) involves {\em decoding to long codes}, in which if the $z$ variables pass the test (with high probability) then the conclusion is that there are long codes that are close (in Hamming distance) to the vector of values of the $z$ variables. (There might be more than one such long code, depending on the error probability that we allow in the long code test, but their number is small.)
The other version which is the one that leads to tight hardness of approximation results (as in Hastad \cite{H97}) involves {\em decoding to linear combinations}, in which if the $z$ variables pass the test (with high probability) then the conclusion is that there are words that are close to the vector of values of the $z$ variables, where these words are linear combinations are over a small number of long codes words. (Also here there might be several such words that are linear combinations.) Folding over $h$ (and also conditioning over $h$) ensures that every decoded long code in the first version corresponds to an assignment to the $x$ variables that satisfies $h$. However, for the second version, there is a distinction between folding and conditioning. Conditioning over $h$ ensures that every decoded linear combination is over long code words that correspond to assignments to the $x$ variables that satisfies $h$. Folding over $h$ only ensures that at least one long code word does.
%The proofs by Bellare et al \cite{BGS98} use folding over $h$, whereas the proofs
%by Hastad \cite{H97} use conditioning over $h$. 

We remark that on top of folding (or conditioning) over $h$, long code tests employ also folding over true (which regardless of other folding or conditioning operations, drops half of the previously remaining $z$ variables). This is relevant to our discussion, as it illustrates the principle that two $z$ variables can be deemed equivalent not only if they have equal values in all long codes of interest (e.g., all long codes that correspond to assignments to the $x$ variables that satisfy the predicate $h$), but also if they have complementary values in all these long codes.

In order to obtain universal factor graphs, we need that even
if $h$ is changed due to a change of polarity of variables (e.g.,
from $x_{1}\vee x_{2}\vee x_{3}$ to $x_{1}\vee\bar{x}_{2}\vee x_{3}$),
the factor graph does not change. However, changing $h$ changes the
pairing in the folding and also changes the equivalence classes in
the conditioning, and hence changes which $z$ variables remain. As a consequence,
the resulting factor graphs change.

Our solution to this problem is through the notion of an \emph{oblivious
folding}.  Consider for example a predicate $h$ corresponding to clause $C_j  = x_{1} \vee x_{2} \vee x_{3}$. The long code associated with an assignment to the three variables $x_{1}, x_{2}, x_{3}$ will have $2^{2^3}$ variables ($z$ variables). Folding over $h$ will pair some of these variables. If polarities in $C_j$ change (say, to $x_{1}\vee\bar{x}_{2}\vee x_{3}$) we get a different predicate $h'$ that leads to a different pairing. To overcome this problem, oblivious folding introduces auxiliary variables and corresponding new predicates. In the example above, this entails introducing three fresh auxiliary variables $y_{j1}, y_{j2}, y_{j3}$ and replacing the clause $C_j$ by the conjunction of four constraints, one that we call here a {\em shadow constraint} $y_{j1} \vee  y_{j2}\vee y_{j3}$, and three {\em equality constraints} $y_{j1} = x_1$,  $y_{j2} = x_2$, $y_{j3} = x_3$.
If polarities in $C_j$ change as above, the shadow constraint remains unchanged and only the equality constraints change
(to $y_{j2} \not= x_2$ in our example). Note that an equality constraint is simply negated when the corresponding $x$ variable is negated (equality changes to inequality).
Now consider a long code associated with an assignment to six variables (namely, to the $x$ variables and the auxiliary $y$ variables). The number of $z$ variables is now $2^{2^6}$.
Consider folding over an equality constraint $h'$. If the constraint is negated to give an inequality constraint $\bar{h'}$, then
rather than change the pairing in the folding (to be folding over $\bar{h'}$ rather than over $h'$), one can instead keep
the same folding as for the original $h'$, but flip the nature of
the relation between the two $z$ variables that form a pair (instead of
requiring them to be different, requiring them to be the same).

In summary, our oblivious folding is performed as follows. Given a set of $x$ variables and a predicate $h$, we
add auxiliary $y$ variables and express $h$ as the conjunction of a shadow predicate over the $y$ variables and equality constraints between the $x$ and $y$ variables. We consider the long code with respect to the combination of all variables. Thereafter we fold
over all the new predicates one by one (the order does not matter),
eventually giving a partition of the $z$ variables into equivalent
classes. In fact, for the shadow constraint we could use conditioning rather than folding -- the main aspect is that for each of the equality constraints we use folding. If the polarity of an $x$ variable changes, the partition of $z$ variables does not change. The only change is the relation among $z$ variables within the partition -- two variables that were previously deemed equal might now be considered as negations of each other. As a result, when an $x$ variable is negated the factor graph over the $z$ variables remains unchanged, and only polarities of some of the $z$ variables change.

Our oblivious folding can be used in conjunction with the proof of Bellare et al~\cite{BGS98}
(after proper modifications), because our equivalence classes
capture all equivalences used by (standard) folding over $h$. We do not know whether our oblivious folding can be used in conjunction
with the proof of Hastad~\cite{H97}, because oblivious folding does not capture all equivalences captured by conditioning. For example, given a long code (of length $2^{2^t}$) for $t$ variables, conditioning over an equality constraint between two variables creates $2^{2^{t-1}}$ equivalence classes, whereas folding (which is an oblivious folding, since we only have an equality constraint) creates $2^{2^t - 1}$ equivalence classes. As explained earlier, when decoding to a linear combination of long codes, we know that for each equality constraint at least one long code in the linear combination satisfies it. However, several equality constraints are introduced by our oblivious folding, and it might be the case that none of the long codes in the linear combination satisfy all of them.

%Hastad's long
%code test (unlike the one of Bellare et al~\cite{BGS98}) does not accept only long codes, but also linear combinations
%of long codes. Using conditioning, Hastad could enforce that
%the only long codes participating in these linear combinations are
%those that satisfy the predicate $h$. With oblivious folding, this is no longer guaranteed. The linear combination may include an odd number
%of long codes such that for some new equality constraint introduced by
%our oblivious folding, an even number of the long codes do not satisfy it.

\subsection{Proof of Theorem~\ref{thm:threshold1}}

Recall that the FGPR in Lemma \ref{AlphabetReduction} creates a 1-restricted
constraint hypergraph of rank 4, such that the degree of every vertex
depends only on the size of the alphabet and the degree of the constraint
graph. Additionally, the FGPR in Theorem \ref{GapAmplification} creates
a bounded degree graph with alphabet that depends only on $t$. Thus,
it is possible to use the series of FGPRs in Theorem \ref{DoublingGap}
without the FGPR in Lemma \ref{Hypergraph2Formula}, to get a universal
factor graph for 1-restricted constraint hypergraphs of rank 4, such
that the degree of every vertex is bounded by a universal constant.

We now follow the proof of Bellare et al.~\cite{BGS98} (with the modification of
the folding) to show hardness of approximation within a factor of
$\frac{77}{80}$.

Given some fixed 1-restricted constraint hypergraph $H$ with alphabet
$\left\{ 0,1\right\} $, rank 4 and bounded degree, consider the two
prover game $H^{u}$: The outer verifier picks uniformly at random
a set $C$ of $u$ hyperedges (constraints). For a constraint $c$,
let $B_{c}$ be the set of bits the value of $c$ depends on. For
every $c\in C$, the outer verifier picks uniformly at random a bit
in $B_{c}$. Call the set of selected bits $B$. Then, the outer verifier
gives $P_{1}$ the set $C$ and gives $P_{2}$ the set $B$. $P_{1}$
is expected to return a satisfying assignment to each edge. $P_{2}$
is expected to return an assignment to the set of bits, consistent
with the assignment $P_{1}$ gives (the same assignment for the same
bits). The outer verifier accepts if both conditions on the responses
of the provers hold.

For $u=1$, it is easy to see that the answers of $P_{2}$ define
an assignment, $a$, for the graph and the answers of $P_{1}$ define
an assignment for each edge, separately. If an assignment for some
edge is inconsistent with $a$ and the verifier chose this edge, the
probability that the verifier finds the inconsistency is at least
$\frac{1}{4}$ (the probability that $P_{2}$ is asked to reveal an
inconsistent bit). Thus, if $\mathrm{UNSAT}\left(G\right)>\epsilon$,
the verifier rejects with probability at least $\frac{\epsilon}{4}$.
It is immediate that if $\mathrm{UNSAT}\left(G\right)=0$, there is
an assignment that makes the verifier always accept.

Using the Parallel Repetition Theorem \cite{R98}, for every $\epsilon$,
if the original graph is at most $1-\epsilon$ satisfiable, there
is $c_{\epsilon}>0$ such that the verifier $V^{u}$ accepts with
probability at most $c_{\epsilon}^{u}$, for all $u>0$. Again, if
the graph is satisfiable, the verifier (after parallel repetition)
can be made to always accept.
\begin{definition}
The string $A\in\left\{ 0,1\right\} ^{\mathcal{F}_{n}}$ (where $\mathcal{F}_{n}$
is the set of all functions from $n$ bits to 1 bit) is the \emph{long
code} of a string $x\in\left\{ 0,1\right\} ^{n}$ if for all $f\in\mathcal{F}_{n}$,
$A_{f}=f\left(x\right)$.
\end{definition}

\begin{definition}
$A\in\left\{ 0,1\right\} ^{\mathcal{F}_{n}}$ is said to be \emph{folded
over $\left(h,b\right)$} ($h\in\mathcal{F}_{n},b\in\left\{ 0,1\right\} $),
if for all $f\in\mathcal{F}_{n}$ $A_{f+h}=A_{f}+b$.

Let $h_{1},\ldots,h_{k}$ be linearly independent functions and $b_{1},\ldots,b_{k}$
be bits. Let $\prec$ be some total ordering of $\mathcal{F}_{n}$.
Let $\mu\left(f\right)$ be the minimal function among $\left\{ f+\sum\sigma_{i}h_{i}|\sigma_{i}\in\left\{ 0,1\right\} \right\} $,
and let $\sigma_{i}^{f}$ be such that $\mu\left(f\right)=f+\sum\sigma_{i}^{f}h_{i}$.
We say that $B\in\left\{ 0,1\right\} ^{\mathcal{F}_{n}}$ is the \emph{folding
over} $\left(h_{1},\ldots,h_{k}\right),\left(b_{1},\ldots,b_{k}\right)$
of $A\in\left\{ 0,1\right\} ^{\mathcal{F}_{n}}$ if $B_{f}=A_{\mu\left(f\right)}+\sum\sigma_{i}^{f}b_{i}$.
Note that $B$ is folded over \emph{$\left(h_{i},b_{i}\right)$.}
\end{definition}

It is possible  to define folding over a set of linearly dependent
functions, provided that the values that the values of the bits $\{b_i\}$ are consistent with the linear dependencies over the functions $\{h_i\}$.  We do not present such a definition, because it is not required for our proofs.

%Suppose that $h_{k+1}$ is linearly dependent on $h_{1},\ldots,h_{k}$,
%and that $h_{1},\ldots,h_{k}$ are linearly independent. Then the
%folding over $\left(h_{1},\ldots,h_{k},h_{k+1}\right),\left(b_{1},\ldots,b_{k},0\right)$
%and $\left(h_{1},\ldots,h_{k},h_{k+1}\right),\left(b_{1},\ldots,b_{k},1\right)$
%are the same, which is impossible. Actually, in this case, there is
%a $b_{k+1}$ such that no string can be folded over $\left(h_{i},b_{i}\right)$
%for all $1\leq i\leq k+1$.

For our usage of folding, it suffices to fold over linearly independent
functions.

In order to check the validity and consistency of the assignment using
disjunctive clauses with 3 literals each, the provers are expected
to give the long code for the assignment for the variables they are
given (the sets $C$ and $B$ as before), and an inner verifier will
check their answer. Let $A$ be the answer of $P_{1}$, $D$ the answer
of $P_{2}$. Let $h_{1},\ldots,h_{u}$ be the homogeneous part of
the constraints in $C$ (the constraints are degree two polynomials).
The verifier will access the folding of $A$ over $\left(1,h_{1},\ldots,h_{u}\right),\left(1,b_{1},\ldots,b_{u}\right),$where
$b_{1},\ldots,b_{u}$ are the constants of the respective constraints.
Since every vertex appears a bounded number of times, independent
of the size of the graph, for large enough hypergraphs the functions
will be linearly independent with very high probability. If they are
dependent, we can always accept while only slightly increasing the
satisfiability of the game.

Note that when we say that the verifier checks the bit $A_{f}$, it
will actually access the bit $A_{\mu\left(f\right)}$ and will invert
it depending on $\left(h_{1},\ldots,h_{u}\right),\left(b_{1},\ldots,b_{u}\right)$.
$\mu\left(f\right)$ only depends on $\left(h_{1},\ldots,h_{u}\right)$.
Close constraint hypergraphs have the same constraints, up to a constant.
So, if two constraint hypergraphs are close, the accessed bit will
be the same and the only difference will be in whether the bit is
negated or not.

The verifier will be modeled as a 3CNF formula, such that all close
constraint hypergraph are transformed into a formula with the same
factor graph. The inner verifier will check one of the following four
constraints:
\begin{enumerate}
\item For $f,g\in\mathcal{F}_{4u}$ chosen uniformly, check that $A_{f}+A_{g}=A_{f+g}$.
This test passes iff the four clauses $\overline{A_{f}}\vee\overline{A_{g}}\vee\overline{A_{f+g}},\overline{A_{f}}\vee A_{g}\vee A_{f+g},A_{f}\vee\overline{A_{g}}\vee A_{f+g},A_{f}\vee A_{g}\vee\overline{A_{f+g}}$
are satisfied.
\item For $f,g,h\in\mathcal{F}_{4u}$ chosen uniformly, if $A_{f}=0$, check
that $A_{fg+h}=A_{h}$. If $A_{f}=1$, check that $A_{fg+g+h}=A_{h}$.
This test passes iff the four clauses $A_{f}\vee\overline{A_{fg+h}}\vee A_{h},A_{f}\vee A_{fg+h}\vee\overline{A_{h}},\overline{A_{f}}\vee\overline{A_{fg+g+h}}\vee A_{h},\overline{A_{f}}\vee A_{fg+g+h}\vee\overline{A_{h}}$
are satisfied.
\item For $f\in\mathcal{F}_{4u},g'\in\mathcal{F}_{u}$ chosen uniformly,
check that $D_{g'}=A_{g+f}+A_{f}$ (where $g$ is $g'$ extended to
the domain of all $3u$ bits of $C$, such that $g$ does not depend
on the bits in $C$ that are not in $B$). This test passes iff the
four clauses $\overline{A_{f}}\vee\overline{A_{g+f}}\vee\overline{D_{g'}},\overline{A_{f}}\vee A_{g+f}\vee D_{g'},A_{f}\vee\overline{A_{g+f}}\vee D_{g'},A_{f}\vee A_{g+f}\vee\overline{D_{g'}}$
are satisfied.
\end{enumerate}
The verifier chooses which of the constraints to check with some probability
that will be implied in the proof.
\begin{definition}
The distance between two strings $x,y\in\left\{ 0,1\right\} ^{\ell}$
is the fraction of coordinates in which they differ.\end{definition}
\begin{claim}
\label{claim:1/2-far}Let $E\in\left\{ 0,1\right\} ^{\mathcal{F}_{n}}$
be a long code of some string $x\in\left\{ 0,1\right\} ^{n}$. Let
$D\in\left\{ 0,1\right\} ^{\mathcal{F}_{n}}$ be some string folded
over $\left(h,b\right)$ such that $h\left(x\right)\neq b$. Then
$D$ and $E$ are $\nicefrac{1}{2}$-far.\end{claim}
\begin{proof}
$\alpha=\left\{ f|E_{f}=D_{f}\right\} $. If $\left|\alpha\right|>\nicefrac{1}{2}\left|\mathcal{F}_{n}\right|$,
then there is $f$ such that $f,f+h\in\alpha$. $E_{f}=f\left(x\right)$,
$E_{f+h}=f\left(x\right)+h\left(x\right)=f\left(x\right)+b+1$. However
$D_{f+h}=D_{f}+b$, so it is impossible that $f,f+h\in\alpha$.\end{proof}
\begin{definition}
Fix $\delta>0$ arbitrarily small. We say that $A,D$ are \emph{$\left(C,B,\delta\right)$-solid}
if $A,D$ are $\left(\nicefrac{1}{2}-\delta\right)$-close to some
$\tilde{A},\tilde{D},$ respectively, where $\tilde{A}$ is a long
code of an assignment $a$ satisfying the edges in $C$, and $\tilde{D}$
is the long code of an assignment $d$ for $B$, such that $d$ is
consistent with $a$. Usually, $\delta$ will be known from the context
and will be omitted.\end{definition}
\begin{claim}
For a game $G^{u}$, let $E'$ be the set of random coins such that
given the corresponding set $C,B$ to $P_{1},P_{2}$ respectively,
the answers $A,D$ are $\left(C,B\right)$-solid. Then, there are
provers $P_{1},P_{2}$ for the game $G^{u}$ and a set $E$ of random
coins, $\left|E\right|\geq\frac{\delta^{4}}{16}\left|E'\right|$,
such that $V_{u}$ will accept for all choices of random coins in
$E$.\end{claim}
\begin{proof}
Consider a bipartite graph where the vertices of one side are all
possible sets of $u$ constraints queried in $G^{u}$ and the other
side is the set of all possible sets of variables queried. There is
an edge between a set of constraints and a set of variables if there
is a choice of random coins such that the sets are queried at the
same time. Note that the graph produced can be seen as a constraint
graph, representing the game $G^{u}$ with the property that an assignment
satisfying $k$ edges will satisfy $V_{u}$ for $k$ choices of random
coins.

If a string is $\nicefrac{1}{2}-\delta$ close to a long code, then
there are at most $4\delta^{-2}$ long codes close to it (Lemma 3.11
in \cite{BGS98}). From claim \ref{claim:1/2-far}, all of these long
codes must satisfy the selected constraints.

For any edge $\left(C,B\right)$ and respective answers of $\left(P_{1},P_{2}\right)$,
with $\left(A,D\right)$ $\left(C,B\right)$-solid, there are $16\delta^{-4}$
possible choices of assignments for $C$ and $B$ (derived from the
closest $4\delta^{-2}$ for each of them), with at least one choice
satisfying the edge between them. Selecting an assignment randomly
satisfies $16\delta^{-4}\left|E'\right|$ edges in expectation. Thus,
there is an assignment for $G^{u}$ that satisfies at least $\frac{\delta^{4}}{16}\left|E'\right|$
edges.\end{proof}
\begin{corollary}
If $G^{u}$ is at most $\epsilon$ satisfiable, then for any pair
of provers at most $16\delta^{-4}\epsilon$ of the answers can be
solid for their respective queries.
\end{corollary}
Given a query to the provers, let $A,D$ be the answers of $P_{1},P_{2}$,
respectively. Let $x$ be the distance of $A$ from the closest linear
function, $\tilde{A}$.

The fraction of tests (of the first type) that will fail is lower
bounded by the function (\cite{BCHKS95}, also stated as Lemma 3.15
in \cite{BGS98})
\[
\Delta\left(x\right)=\left\{ \begin{array}{cc}
3x-6x^{2} & 0\leq x\leq\frac{5}{16}\\
\frac{45}{128} & \frac{5}{16}<x\leq\frac{45}{128}\\
x & \frac{45}{128}<x\end{array}\right.
\]

If $\tilde{A}$ is not a long code, then the fraction of tests that
the second test will fail is at least $\nicefrac{3}{8}\left(1-2x\right)$
(Lemma 3.19 in \cite{BGS98}).

If $\tilde{A}$ is a long code of $a$, $\tilde{D}$ is a long code
of $d$, where $d$ is the restriction of $a$ on the respective bits,
and $y$ is the distance between $D$ and $\tilde{D}$, then the fraction
of tests that the third test will fail is at least $y\left(1-2x\right)$
(Lemma 3.21 in \cite{BGS98}).
\begin{theorem}
\label{BestGap3SAT}There is a $\left(\frac{77}{80}+\epsilon\right)$-universal
factor graph for 3-SAT, for any $\epsilon>0$.\end{theorem}
\begin{proof}
Let $\mathcal{G}_{n}$ be the set of graphs of size $n$ produced
by the transformation from lemma \ref{Formula2Graph} activated on
APX-universal factor graph generated in appendix \ref{app:APXUFG}.
Note that the degree of every vertex is bounded due to the transformation
used to create the factor graph. The factor graph of the 3CNF formula
checking the satisfiability of $G^{u}$ is the same, for all $G\in\mathcal{G}_{n}$.

Suppose that $G$ cannot be completely satisfied. Then, the fraction
of answers that can be solid to their respective queries is arbitrarily
small (fixing small $\delta$, increasing $u$). Non solid answers
fail some of the tests, so our goal is to increase the fraction of
clauses that non solid answers cannot satisfy. If $A,D$ is the answer
to the query $\left(C,B\right)$ which is not $\left(C,B\right)$-solid
then there are several possibilities:
\begin{itemize}
\item $\tilde{A}$ is not a long code, then at least $n_{1}\Delta\left(x\right)+\frac{3n_{2}}{8}\left(1-2x\right)$
tests fail.
\item $\tilde{A}$ is the long code of $a$, but $D$ is at least $\nicefrac{1}{2}-\delta$
far from a long code of the restriction of $a$, then at least $n_{1}\Delta\left(x\right)+n_{3}\left(\nicefrac{1}{2}-\delta\right)\left(1-2x\right)$
tests fail.
\end{itemize}
Otherwise, $A,D$ is $\left(C,B\right)$-solid.

Let $n_{1},n_{2},n_{3}$ be the number of tests of each type for every
query (theses numbers uniquely define the probability the inner verifier
chooses which of the three tests to execute). The total number of
clauses is $4\left(n_{1}+n_{2}+n_{3}\right)$. If a test fails then
at least a single clause cannot be satisfied (using the respective
assignment). For $n_{3}=\frac{3}{4}n_{2}$ all the fractions of failed
tests (and the number of unsatisfied clauses) is the same for both
cases ($\delta$ can be arbitrarily small). Let $k$ be the total
number of clauses. We need to maximize the fraction of unsatisfied
clauses, that is \[
\frac{n_{1}\Delta\left(x\right)+\frac{3}{56}\left(k-4n_{1}\right)\left(1-2x\right)}{k}\]
Then, the minimum must be at $x=0$, $x=\frac{45}{128}$, or $x=\nicefrac{1}{2}$,
so $n_{1}=n_{3}=\frac{3k}{40}$, $n_{2}=\frac{k}{10}$, which gives
that at least a $\frac{3}{80}$ fraction of the clauses are unsatisfiable.
\end{proof}

\section{A Remark on Nearly Tight Thresholds}
\label{sec:tight}

The following remark relates to the proof of Theorem \ref{thm:tight}.

\begin{remark}
\label{Weighted2Unweighted} If $\phi_3$ is unweighted, then weighted formula $\phi_{k}$ can easily
be replaced by an unweighted formula without affecting the correctness
of the proof. Scale all weights by a multiplicative constant so that
clauses in $\psi_{0}$ have weight 1, and then clauses in $\psi_{i}$
for $i\ge1$ have weight $\gamma m$. We may assume without loss of
generality that $\gamma m$ is an integer. (Otherwise, decrease $\gamma$
by a little.) Duplicating each formula $\psi_{i}$ (for $i\ge1$)
$\gamma m$ times, each time using a fresh set of new three variables
as the $z$ variables, gives the desired unweighted version of the
formula $\phi_{k}$.
\end{remark}

\section{\label{app:MoreCSPs}Universal Factor Graphs for Additional CSPs}

Many known reductions are in fact FGPRs, and this gives universal
factor graphs for additional CSPs. Examples include the standard reductions from max-3SAT to max-4NAE (adding a variable to all clauses), from max-4NAE to max-3NAE (break each clause
into two clauses of two literals and add a new variable to one
clause and its negation to the other) and from max-3NAE to max-2LIN (replace every 3NAE clause by three 2LIN clauses, each of which is the XOR of two literals from the NAE clause).

\end{document}